\documentclass[DIV=15,11pt,abstract]{scrartcl}

\title{Duality of Orthogonal and Symplectic
Random Tensor Models: General Invariants }
\date{}
\usepackage{authblk}
\author[1]{Hannes Keppler}
\affil[1]{\normalsize\itshape 
	Heidelberg University, Institut f{\"u}r Theoretische Physik, Philosophenweg 19, 69120 Heidelberg, Germany
	\authorcr \hfill}
\author[2]{Thomas Muller}
\affil[2]{\normalsize\itshape 
	Université de Bordeaux, LaBRI CNRS UMR 5800, Talence, France
	\authorcr \hfill}

\usepackage[T1]{fontenc}
\usepackage[utf8]{inputenc}
\usepackage[english]{babel}

\usepackage{lmodern}
\usepackage{microtype}

\setkomafont{title}{\large\rmfamily\bfseries}
\setkomafont{section}{\normalfont\rmfamily\bfseries\large}
\setkomafont{subsection}{\normalfont\rmfamily\bfseries}
\setkomafont{subsubsection}{\normalfont\rmfamily\bfseries\normalsize}
\setkomafont{paragraph}{\normalfont\rmfamily\bfseries}
\setkomafont{subparagraph}{\normalfont\itshape}

\RedeclareSectionCommands[    
tocentryformat=\normalfont\bfseries,
tocpagenumberformat=\normalfont\bfseries
]{section}
\RedeclareSectionCommands[    
tocentryformat=\normalfont,
tocpagenumberformat=\normalfont
]{subsection,subsubsection,paragraph,subparagraph}

\usepackage[intlimits,sumlimits]{amsmath}
\usepackage{amssymb, commath, mathtools, bm, amsthm, thmtools, relsize}
\usepackage[mathscr]{eucal}    
\usepackage{dsfont}
\usepackage{tensor}
\usepackage{cancel}

\numberwithin{equation}{section}

\usepackage{tikz}
\usetikzlibrary{tikzmark, cd, positioning}

\usepackage{graphicx} 
\usepackage{blindtext} 
\usepackage[format=plain, labelfont=bf, labelsep=period]{caption} 

\usepackage[bottom]{footmisc}   

\usepackage[toc,page]{appendix}
\usepackage[noadjust]{cite} 
\usepackage{xcolor}
\definecolor{myBlue}{RGB}{1,1,141}

\usepackage[colorlinks=True]{hyperref}
\hypersetup{allcolors=myBlue}

\newtheoremstyle{plain}
{}{}{\itshape}{}{\bfseries}{.}{.5em}{}

\theoremstyle{plain}    
\newtheorem{theorem}{Theorem}
\newtheorem{lem}[theorem]{Lemma}
\newtheorem{cor}[theorem]{Corollary}
\newtheorem{prop}[theorem]{Proposition}
\newtheorem{definition}[theorem]{Definition}

\newcommand{\cA}{\mathcal{A}}
\newcommand{\cB}{\mathcal{B}}

\newcommand{\cD}{\mathcal{D}}

\newcommand{\cG}{\mathcal{G}}

\newcommand{\cP}{\mathcal{P}}

\newcommand{\cS}{\mathcal{S}}

\newcommand{\vP}{\vec{P}}
\newcommand{\vE}{\vec{E}}

\newcommand{\vcB}{\vec{\mathcal{B}}}
\newcommand{\vcG}{\vec{\mathcal{G}}}
\newcommand{\vcA}{\vec{\mathcal{A}}}

\renewcommand{\del}{\partial}

\DeclareMathOperator{\sgn}{sgn}

\begin{document}

{\hypersetup{allcolors=black}
	\maketitle
	\begin{abstract}
	    In Gurau and Keppler 2022 (arXiv:2207.01993), a relation between orthogonal and symplectic tensor models with quartic interactions was proven. In this paper, we provide an alternative proof that extends to polynomial interactions of arbitrary order. We consider tensor models of order $D$ with no symmetry under permutation of the indices that transform in the tensor product  of $D$ fundamental representations of $O(N)$ and $Sp(N)$. We explicitly show that the models obey the $N$ to $-N$ duality graph by graph in perturbation theory.
	\end{abstract}
	\microtypesetup{protrusion=false}
	\setcounter{tocdepth}{2}
	\tableofcontents
	\microtypesetup{protrusion=true}
}

\section{Introduction}

Random tensor models \cite{Ambjorn,gurau,Gurau-inviation,Guraureview,tanasabook,gurau2019notes}, introduced as a generalization of random matrix models, are probability measures of the type:
\begin{equation}
d\mu[T] = e^{-S[T]} \prod_{(a_1,\dots,a_D)}\frac{dT^{a_1\dots a_D}}{\sqrt{2\pi}} \; ,
\end{equation}
where the action \(S[T]\) is build out of invariants under some symmetry transformation. These models are analog to zero-dimensional quantum field theory and their perturbative expansion can be reorganized as a series in $1/N$ \cite{Gurau-N,Bonzom:2012hw,Carrozza:2015adg,sabine,sylvan,Carrozza:2021qos,Dartois_2013,krajewski2023double}. As their Feynman graphs are dual to higher dimensional triangulations,
random tensors provide a framework for the study of random topological spaces; in one dimension tensor models provide an alternative to the Sachdev-Ye-Kitaev model without quenched disorder \cite{Witten:2016iux,Klebanov}; in higher dimensions they lead to tensor field theories and a new class of large  $N$ \emph{melonic} conformal field theories \cite{Giombi:2017dtl,Bulycheva:2017ilt,Giombi:2018qgp,Klebanov:2018fzb,Gurau-TFT}.

In this paper, we study tensor models with symplectic and/or orthogonal symmetry. Several incarnations of the relation between the orthogonal and symplectic group for negative dimensions have been studied in the literature:
On one hand, in the context of representation theory, one can make sense of the relation \mbox{$SO(-N)\simeq Sp(N)$} \cite{King,Cvitanovic,Mkrtchyan-Veselov,Cvitanovicbook}.
On the other hand, for even $N$, $SO(N)$ and $Sp(N)$ gauge theories are known to be related by changing $N$ to $-N$ \cite{Mkrtchian}. A vector model with symplectic fermions in three space-time dimensions has been studied in \cite{LeClair} and an example of $SO(N)$ and $Sp(N)$ gauge theories with matter fields and Yukawa interactions can be found in \cite{Litim}. This duality has furthermore been shown to hold between orthogonal and symplectic matrix ensembles (the $\beta=1,4$ ensembles) \cite{goe-gse}.

The orthogonal/symplectic duality has already been studied for tensor models by one of the authors in \cite{Duality_Hannes}. There, a graded colored tensor model (reviewed in Def.~\ref{def: intro} below) was introduced. It was then shown that the partition function and connected two point correlation function of this model was invariant when replacing $N_c\leftrightarrow-N_c$ and at the same time changing the symmetry $O(N_c)\leftrightarrow Sp(N_c)$. However, the analysis in \cite{Duality_Hannes} made use of an intermediate field/Hubbard-Stratonovich transformation. This method allows to work with bosonic fields only, but is only applicable to the case of quartic interactions. Working directly in the usual colored graph representation of tensor models we generalize results of \cite{Duality_Hannes} to interactions of arbitrary order and proceed in a more direct way.

\paragraph{Main result.}
We consider tensors of order $D$ with no symmetry under permutation of their indices and call the position of an index its color $c$, with $c=1,2,\dots D$. 
The tensors transform in the tensor product of $D$ fundamental representations of $O(N)$ and/or  $Sp(N)$, i.e.~each tensor index is transformed by a different $O(N)$ or $Sp(N)$ matrix.
The tensor components are real fermionic (anticommuting, odd) if the number of $Sp(N)$ factors is odd and real bosonic (commuting, even) if this number is even. It is convenient to assign a parity to the tensor indices: $\abs{c}= 0$ or $\abs{c} = 1$ if the index transforms under $O(N_c)$ or $Sp(N_c)$, respectively.
In this paper, we generalize the results of \cite{Duality_Hannes} by allowing arbitrary polynomial interactions.

\begin{definition}\label{def: intro}
The real graded tensor model, obeys the symmetry:
\begin{equation}
\pmb{O}_1(N_1)\otimes \pmb{O}_2(N_2)\otimes\dots\otimes \pmb{O}_D(N_D), \quad \pmb{O}_c(N_c) = \smash[b]{\begin{cases}
O(N_c), \quad \abs{c}=0 \\
Sp(N_c), \quad \abs{c}=1
\end{cases}} \; ,
\end{equation}
(therefore the name ``graded'') is defined by the measure:
\begin{align}
d\mu[T] \simeq e^{-S[T]}\ \prod_{a_1,\dots,a_D} dT^{a_1\dots a_D} \; ,
 \qquad
S[T]= \frac{1}{2} \Big(T^{a_1\dots a_D} T^{b_1\dots b_D} \prod_{c=1}^D g^c_{a_c b_c}\Big) + \sum_{\substack{\cB\ \text{connected,}\\ |V(\cB)|>2}} \frac{\lambda_{\cB}}{|V(\cB)|} I_{\cB}(T) \; ,
\end{align}
where $g^c_{a_cb_c}$ is the Kronecker $\delta_{a_cb_c} $ for $\abs{c}=0$ or the canonical symplectic form $\omega_{a_cb_c}$ for $\abs{c}=1$ and the sum runs over independent connected trace invariants $I_{\cB}(T)$ of order higher than two, indexed by undirected colored graphs $\cB$ (see Sec.\,\ref{sec: def} for more details).
\end{definition}

The partition function $Z$ and the expectation value of an invariant $\langle{I_{\cB}(T)}\rangle$ are defined by:
\begin{equation}\label{eq: ZandI}
Z(\{\lambda\})=\int d\mu[T], \quad\text{and}\quad \langle{I_{\cB}}(T)\rangle(\{\lambda\} )=\frac{1}{Z} \int d\mu[T]\ I_{\cB}(T) \; ,
\end{equation}
and can be evaluated in a perturbative expansion.
Our main theorem is the following:
\begin{theorem}\label{thm: main}
The perturbative series of the partition function $Z$ and expectation values of invariants $\langle{I_{\cB}(T)}\rangle$ can be expressed as a formal sum over $(D+1)$-colored undirected graphs $\cG$.
Each summand, corresponding to a specific graph $\cG$, writes as a product:
\begin{equation}
K(\{\lambda\},\cG)\cdot \prod_{c\in\cD} \big((-1)^{|c|}N_c\big)^{F_{c/0}(\cG)} \;,
\end{equation}
of a term $K$, encoding the dependence on the coupling constants $\lambda_\cB$ and some combinatorial numbers associated to $\cG$, and a term depending on $N_1,N_2\dots ,N_D$ (see Sec.~\ref{sec: def} for the relevant definitions and Sec.~\ref{sec: pertexp} for the precise form of the series).
\end{theorem}
\begin{proof}
The theorem follows from Prop.~\ref{prop: Z} and Cor.~\ref{cor: I}.
\end{proof}
The essential remark is that all the factors $N_c$ come in the form $(-1)^{\abs c} N_c$, hence  each term is mapped into itself by exchanging $O(N_c) \leftrightarrow  Sp(N_c) $ and $N_c \leftrightarrow -N_c$.
Because graphically each $N_c$ is associated to a face of colors $c/0$ (cycle of edges of alternating colors $c$ and $0$),  this result can be seen as a generalization of the usual minus sign in quantum field theory for each fermionic loop. But one should keep in mind, that the full tensor is not necessary fermionic (its components are not necessarily anticommuting Graßmann numbers).


\section{Setup of the Models}\label{sec: def}
In this section we define the model. First, we specify the space of tensors we are interested in. Second, we give a description of the possible tensor invariants in terms of directed edge colored graphs, and third, we specify the model and its invariance properties.

\paragraph{The tensors.}
A generic tensor $T^{a_1\dots a_D}$ has no symmetry properties under permutation of its indices hence the indices have a well defined position $c$, called their \emph{color}. The set of colors is denoted $\mathcal{D}=\{1,\dots, D\}$. We assign a parity to each color and sometimes call the colors with $\abs{c}=0$ even and the ones with $\abs{c}=1$ odd. The tensor components shall be bosonic (even) if the number of colors with $\abs{c}=1$ (i.\,e.\ odd colors) is even and fermionic (odd) otherwise: the Graßmann number $T^{a_1\dots a_D}$ has the same parity as $\sum_{c\in\mathcal{D}} \abs{c}$.

Let ${H}_c  =  \mathds{R}^{N_c|0}$ for $\abs{c}=0 $, respectively $ {H}_c  =\mathds{R}^{0|N_c}$ for $\abs{c}=1  $ be a real supervector space of dimension \(N_c\) that is either purely even or purely odd. Each $H_c$ is endowed with a non-degenerate \emph{graded symmetric} inner product $g^c$:
\begin{equation}
g^c(u,v) = (-1)^{|c|} g^c(v,u),\quad \forall u, v\in {H}_c\;.
\end{equation}
In a standard basis $g^c$ agrees with the standard symmetric or symplectic form, that is $g^c_{a_cb_c} = \delta_{a_cb_c}$ for $\abs{c}=0 $, respectively $ g^c_{a_cb_c} = \omega_{a_cb_c}$ for $\abs{c}=1 $.
As usual, we write $g^{c,\, a_c b_c}$ for the components of $(g^c)^{-1}$.
The isometry group preserving $g^c$ is either $O(N_c)$ in the $\abs{c}=0$ case or $Sp(N_c)$ in the $\abs{c}=1$ case, denoted collectively by 
$\pmb{O}_{c}(N_c) := \{ O_c\ |\ g^c_{a_cb_c} = O_{a_c}^{\ a'_c} O_{b_c}^{\ b'_c} g^c_{a'_c b'_c}  = ( O g^c O^T )_{a_cb_c}\}$.

The tensors are elements of
\begin{equation}
H_1\otimes H_2\otimes \dots\otimes H_D\;,
\end{equation}
and transform in the tensor product representation of several orthogonal and symplectic groups according to the type of the individual ${H}_c$'s:
\begin{equation}
	T^{a_1\dots a_D} \to \tensor{(O_1)}{^{a_1}_{b_1}} \dots \tensor{(O_D)}{^{a_D}_{b_D}}\ T^{b_1\dots b_D},
	\quad O_1\otimes\dots\otimes O_D\in\bigotimes_{c\in\mathcal{D}} \pmb{O}_c(N_c) \; .
\end{equation}

\paragraph{Directed edge colored graphs and invariants.}
Invariant polynomials in the tensor components are constructed by contracting the indices of color \(c\) with the inner product \(g^c\).
The unique quadratic invariant is:
\begin{equation}
	g^{\otimes D}(T,T)=T^{a_{\cD}} T^{b_{\cD}} \prod_{c\in\cD} g^c_{a_c b_c} \; .
\end{equation}
General \emph{trace invariants} are polynomials in the \(T^{a_{\cD}}\)'s build by contracting pairs of indices of the same color. These invariants admit a graphical representation as \emph{directed edge colored graphs}. 
The graphical representatives of tensor invariants are often called bubbles \cite{gurau,tanasabook}.

\begin{definition}[Directed Edge Colored Graphs]\label{def: coloredgraph}
A closed directed edge \(D\)-colored graph (directed colored graph for short) is a directed graph \mbox{\(\vcB=(V(\vcB),\vE(\vcB))\)} with vertex set \(V(\vcB)\) and edge set \(\vE(\vcB)\) such that:
\begin{itemize}
\item The edge set is partitioned into \(D\) disjoint subsets \mbox{\(\vE(\vcB)=\bigsqcup_{c=1}^{D} \vE^c(\vcB)\)}, where we denote the subset of edges of color \(c\) by \mbox{\(\vE^c(\vcB)\ni e^c=(v,w)\)}, with \(v,w\in V(\vcB)\).
\item Each set $\vE^c(\vcB)$ is a directed pairing of the vertices.
\end{itemize}
\end{definition}
As a consequence all vertices are \(D\)-valent with all the  edges incident to a vertex having distinct colors, and $V(\vcB)$ is of even cardinality. We denote by $F_{c/c'}(\vcB)$ the number of faces of colors $c\neq c'$, that is cycles made of alternating edges of these two distinct colors. Per default, we will consider directed graphs $\vcB$ and view undirected graphs as equivalence classes $\cB =[\vcB]$ of their directed versions. All graphs have labeled vertices. Some examples are depicted in Figs.~\ref{fig: bubbles},\,\ref{fig: interaction}.

\begin{figure}[t]
    \centering
	\begin{tikzpicture}[font=\small]
		\node at (0,0) {\includegraphics[scale=4]{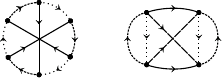}};
	\end{tikzpicture}
	\caption{Examples of colored directed graphs: a 3-colored sextic graph on the left (the wheel) and a 4-colored quartic one on the right. The different colors are represented by different line styles}
	\label{fig: bubbles}
\end{figure}

Due to the signs introduced by the reversing the edges of odd colors (remember that this corresponds to transposing the antisymmetric matrix $\omega$), several invariants differ only by a minus sign. This ambiguity can be fixed by using a sign fixing prescription generalizing the one given in  \cite{Duality_Hannes}.

The invariant of order $2k$ defined by the directed colored graph $\vcB$ with vertices $V(\vcB)=\{1,2,\dots,2k \}$, will be denoted by $I_{\vcB}(T)$, and is given by the following expression:
\begin{equation}\label{eq: Invariants}
\begin{split}
 I_{\vcB}(T) &=   \sum_{a^1_\cD, a^2_\cD, \dots, a^{2k}_\cD} \Big(\prod_{(i,j) \in \vP_{ref,2k} }T^{a^i_\cD} T^{a^j_\cD} \Big) \prod_{c\in\cD} \Big(\epsilon(\vP_{ref,2k}, \vE^c(\vcB))^{|c|} \prod_{(i,j) \in \vE^c(\vcB)}
 g^c_{a^i_c a^j_c}\Big) \\
 &=   \sum_{a^1_\cD, a^2_\cD, \dots, a^{2k}_\cD} \Big(\prod_{(i,j) \in \vP_{ref,2k} }T^{a^i_\cD} T^{a^j_\cD} \Big)  \Big( \prod_{c\in\cD}\epsilon(\vP_{ref,2k}, \vE^c(\vcB))^{|c|} \; K^c_{\vcB,a^1_c,\dots,a^{2k}_c}\Big) \;,
\end{split}
\end{equation}
here $\vP_{ref,2k}$ is an arbitrary but fixed reference pairing on the set of vertices, chosen to be \linebreak[4]\mbox{$\vP_{ref,2k} = \{(1,2),\dots,(2k-1,2k)\}$}, and the sign of the pairings $\epsilon(\vP_{ref,2k_p},  \vE^c(\vcB))^{|c|}$ was introduced to fix the sign ambiguity. As an example, consider the pillow interaction (Fig.~\ref{fig: interaction}) defined by the graph $\vec{\mathcal{P}}$. The corresponding tensor invariant is:
\begin{equation}\label{eq: exInvariant}
    \begin{split}
    I_{\vec{\mathcal{P}}} = (-1)^{|1|+|2|}  \sum_{a,b,c,d} \left( T^{a_1 a_2 a_3} T^{b_1 b_2 b_3} \right) \left( T^{c_1 c_2 c_3} T^{d_1 d_2 d_3} \right) \;
     g^1_{c_1 a_1} \; g^1_{d_1 b_1} \; g^2_{c_2 a_2} \; g^2_{d_2 b_2} \; g^3_{a_3 b_3} \; g^2_{c_3 d_3} \;.
    \end{split}
\end{equation}
Before reviewing the definition and properties of $\epsilon$, which will play a central role, let us comment about how the sign ambiguity shall be understood:
\begin{itemize}
    \item Because the tensors may anticommute, for writing down the expression for the invariant, it is necessary to fix an order. This is done by the reference pairing $\vec{P}_{ref,2k}$.
    \item If two directed colored graphs $\vcA$ and $\vcB$ differ only by redirecting some of their edges, the corresponding trace invariants are the same, up to transposing some $(g^c)^T=(-1)^{|c|}g^c$ which could lead to a sign difference between $I_{\vcA}(T)$ and $I_{\vcB}(T)$. In order to restrict to independent invariants, we consider such directed graphs to be in the same \textit{equivalence class} $[\vcA]=[\vcB]$, i.e.~they describe the same \textit{undirected} colored graph.
    \item The sign prescription ensures that $I_{\vcB}(T)$ is a \textit{class function}:
    \begin{equation} I_{\vcA}(T)= I_{\vcB}(T)\; \quad \text{if}\ \cA=\cB\;, \end{equation}
    and thus any representative of $[\vcB]$ can be used to write down the invariant.
\end{itemize}

Because of the last point, from now on, we will label the invariants by undirected graphs $\cB$, and it is understood that for \eqref{eq: Invariants} an arbitrary directed representative $\vcB\in[\vcB]\equiv\cB$ has been chosen.

A trace invariant is called connected, if the corresponding colored graph is so. Note that any product of two invariants can be written as a single disconnected invariant such that:
\begin{equation}\label{eq: unionofinv}
    I_{\cA}(T) I_{\cB}(T) = I_{\cA\sqcup\cB} (T) \;,
\end{equation}
and a new reference pairing is given by the disjoint union of the original reference pairings \eqref{prop3}.

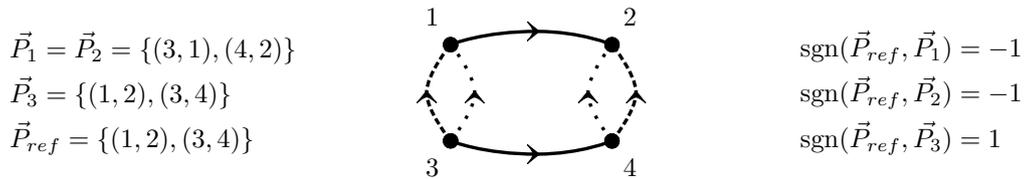
\begin{figure}[t]
    \centering
	\begin{tikzpicture}[font=\small]
		\node at (0,0) {\includegraphics[scale=4]{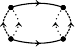}};
		\node at (-1.3,1) {$1$};
		\node at (1.3,1) {$2$};
		\node at (-1.3,-1) {$3$};
		\node at (1.3,-1) {$4$};
		
		\node [anchor=west] at (-7,.6) {\(\displaystyle \vec{P_1}=\vP_2=\lbrace (3,1),(4,2) \rbrace\)};
		\node [anchor=west] at (-7,0) {\(\displaystyle \vec{P_{3}}=\lbrace (1,2),(3,4) \rbrace\)};
		\node [anchor=west] at (-7,-.6) {\(\displaystyle \vec{P}_{ref}=\lbrace (1,2),(3,4) \rbrace\)};
		
		\node at (5,0.6) {\(\displaystyle \sgn(\vec{P}_{ref},\vec{P_1})= -1\)};
		\node at (5,-0.0) {\(\displaystyle \sgn(\vec{P}_{ref},\vec{P_2})= -1\)};
		\node at (5,-0.6) {\(\displaystyle \sgn(\vec{P}_{ref},\vec{P_3})= 1 \hphantom{-}\)};
	\end{tikzpicture}
	\caption{A colored directed graph (the pillow $\vec{\mathcal{P}}$) defined by three pairings. The edges of colors one, two and three are represented by dashed, dotted and solid lines, respectively. The associated invariant is given in \eqref{eq: exInvariant}.}
	\label{fig: interaction}
\end{figure}

\paragraph{Sign of oriented pairings.}

Consider two oriented pairings $\vP_1$ and $\vP_2$ on the same set of $2k$ elements:
\begin{equation}
\begin{aligned}
    \vP_1 &= \{ (i_1, i_2) , \dots , (i_{2k-1}, i_{2k}) \} \;, \\
    \vP_2 &= \{ (j_1, j_2) , \dots , (j_{2k-1}, j_{2k}) \} \;.
\end{aligned}
\end{equation}
The \textit{sign $\epsilon(\vP_1,\vP_2)$ of the two pairings with respect to each another} is defined as the sign of the permutation that takes $i_1\dots i_{2k}$ into $j_1 \dots j_{2k}$. The properties of this sign are: 
\begin{enumerate}
    \item The sign is symmetric under permutation of its arguments:
    \begin{equation}\label{prop1}
        \epsilon(\vP_1,\vP_2) = \epsilon(\vP_2,\vP_1) \;.
    \end{equation}
    \item For three pairings $\vP_1$, $\vP_2$, $\vP_3$ on the same set, one has:
        \begin{equation}\label{prop2}
            \epsilon(\vP_1,\vP_2) = \epsilon(\vP_1,\vP_3) \epsilon(\vP_2,\vP_3) \;.
        \end{equation}
    \item For two pairings $\vP_1$, $\vP_2$ on a first set $\mathcal{S}_1$ of $2k$ elements and two pairings $\vP_3$, $\vP_4$ on a second set $\mathcal{S}_2$ of $2p$ elements, the product $ \epsilon(\vP_1,\vP_2) \epsilon(\vP_3,\vP_4)$ can be written as the sign of the disjoint union of pairings $\vP_1 \sqcup \vP_3$ and $\vP_2 \sqcup \vP_4$ (on the set $\mathcal{S}_1 \sqcup \mathcal{S}_2$):
    \begin{equation}\label{prop3}
       \epsilon(\vP_1,\vP_2) \epsilon(\vP_3,\vP_4) = \epsilon(\vP_1 \sqcup \vP_3,\vP_2 \sqcup \vP_4) \;.
    \end{equation}
\end{enumerate}
The sign of two pairings has a nice graphical interpretation that will be of great use.

\begin{figure}[t]
	\centering
	\begin{tikzpicture}[font=\footnotesize] 
	\node at (0,0) {\includegraphics[scale=3]{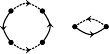}};
	\node at (-1.6,.55) {\(1\)};
	\node at (-1.6,-.6) {\(4\)};
	\node at (-.05,-.6) {\(3\)};
	\node at (-.05,.55) {\(2\)};
	\node at (.55,-.25) {\(5\)};
	\node at (1.75,-.25) {\(6\)};
	\node at (-5,.3) {\(\displaystyle \vec{P_1}=\lbrace (1,4),(3,2),(5,6) \rbrace\)};
	\node at (-5,-.3) {\(\displaystyle \vec{P_2}=\lbrace (1,2),(4,3),(6,5) \rbrace\)};
	\node at (5,0) {\(\displaystyle \sgn(\vec{P_1},\vec{P_2})= (-1)^{F_{1/2,even}} = -1\)};
	\end{tikzpicture}
	\caption{Illustration of Lemma~\ref{lem:sign}. \(\vec{P_1}\) is represented with solid and \(\vec{P_2}\) with dashed edges. The face on the left is odd and the face on the right is even.}
	\label{fig: sign}
\end{figure}
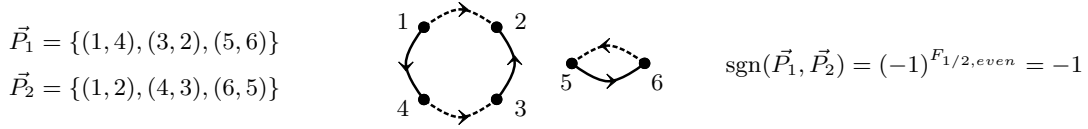
\begin{figure}[t]
    \centering
	\begin{tikzpicture}[font=\footnotesize]
		\node at (0,0) {\includegraphics[scale=3]{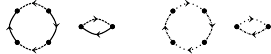}};
		\node at (-3.91,.5) {$1$};
		\node at (-2.35,.5) {$2$};
		\node at (-2.35,-.5) {$3$};
		\node at (-3.91,-.5) {$4$};
		\node at (-1.9,0) {$5$};
		\node at (-.45,0) {$6$};
		\node at (-6.2,.3) {\(\displaystyle \vec{P_c}=\lbrace (2,3),(4,1),(6,5) \rbrace\)};
		\node at (-6.2,-.3) {\(\displaystyle \vec{P_{c'}}=\lbrace (2,1),(3,4),(5,6) \rbrace\)};
		\node at (.85,.5) {$1$};
		\node at (2.4,.5) {$2$};
		\node at (2.4,-.5) {$3$};
		\node at (.85,-.5) {$4$};
		\node at (2.85,0) {$5$};
		\node at (4.3,0) {$6$};
		\node at (6.0,.8) {\(\displaystyle \rho=(1234)(56)\)};
		\node at (6.6,.2) {\(\displaystyle \vec{P_r}=\lbrace (2,3),(4,1),(6,5) \rbrace\)};
		\node at (6.6,-.4) {\(\displaystyle \vec{P_{r'}}=\lbrace (1,2),(3,4),(5,6) \rbrace\)};
	\end{tikzpicture}
	\caption{Illustration of the Proof of Lemma~\ref{lem:sign}. The faces of colors $r/r'$ coincide with the cycles of $\rho$, and with the faces of colors $c/c'$, up to the orientation of their edges. $\vP_c$ solid line, $\vP_{c'}$ dashed, $\vP_r$ dash-dotted, $\vP_{r'}$ dotted.}
	\label{fig: sign_proof}
\end{figure}
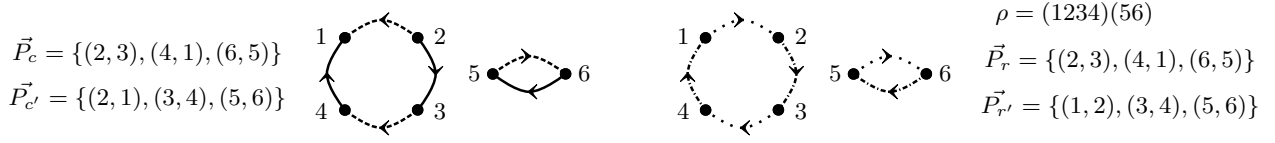

\begin{lem}\label{lem:sign}
Depicting the $2k$ elements of a set $\cS$ as vertices, the two pairings $\vP_{c}$ and $\vP_{c'}$ on this set can be represented by colored (one color -- $c$ or $c'$ -- for each pairing), oriented edges connecting the vertices. Define a face of colors $c/c'$ as an alternating cycle of edges of color $c$ and $c'$. A face is called even (resp.~odd), if an even (odd) number of its edges point in the same directions around its cycles.\footnote{Note that a face of colors $c/c'$ always has an even number of edges, and hence this notion is well defined.} Denoting $F_{c/c',even}$ resp.~$F_{c/c',odd}$ the number of even and odd faces of colors $c/c'$, the sign $\epsilon(\vP_c,\vP_{c'})$ can be expressed as: 
\begin{equation}
    \epsilon(\vP_{c},\vP_{c'}) = (-1)^{F_{c/c',even}} \;.
\end{equation}
\end{lem}
\noindent See Fig.~\ref{fig: sign} for an illustration.
\begin{proof}
Denoting $\vP_c = \{ (i_1, i_2) , \dots , (i_{2k-1}, i_{2k}) \}$ and $\vP_{c'} = \{ (j_1, j_2) , \dots , (j_{2k-1}, j_{2k}) \}$, by definition, $\epsilon(\vP_c,\vP_{c'})$ is the sign of the permutation $\sigma=\big(\begin{smallmatrix}i_1&i_2&\dots&i_{2k-1}&i_{2k}\\ j_1&j_2&\dots&j_{2k-1}&j_{2k}\end{smallmatrix}\big)$.
	One can define a second permutation $\rho$, whose cycles coincide with the faces of $\vP_c \sqcup \vP_{c'}$ (neglecting for a moment the orientation of the edges). 
	The permutation $\rho$ is defined by:
	\begin{itemize}
		\item $\rho (i)$ is the successor of the vertex $i$ that is reached by going clockwise around the face of colors $c/c'$ to which $i$ belongs.
		\item Writing $\rho =\big(\begin{smallmatrix}l_1&l_2&\dots&l_{2k-1}&l_{2k}\\ m_1&m_2&\dots&m_{2k-1}&m_{2k}\end{smallmatrix}\big)$, one can think of $\rho$ as consisting of two directed pairings $\vP_r = \{ (l_1, l_2) , \dots , (l_{2k-1}, l_{2k}) \}$ and $\vP_{r'} = \{ (m_1, m_2), \dots , (m_{2k-1}, m_{2k}) \}$ that differ from $\vP_c$ and $\vP_{c'}$ only by the direction of their edges.
	\end{itemize}
	Since all the faces are of even length, the sign of the permutation $\rho$ is given by $\sgn(\rho) = (-1)^{F_{c/c'}}$. The crucial point is, that $\rho$ is chosen such that it differs from $\sigma$ by an odd number of transpositions for each odd face, i.e.~$\sgn(\sigma) = (-1)^{F_{c/c',odd}} \sgn(\rho) = (-1)^{F_{c/c',even}}$. See Fig.~\ref{fig: sign_proof} for an illustration.
\end{proof}

For a review on the connection of directed pairings and their sign with pfaffians and fermionic Gaußian integrals, we refer the interested reader to the Appendix of \cite{integrals}.

\paragraph{Graded colored tensor model.}
As discussed above, the set of independent trace invariants is indexed by equivalence classes of directed colored graphs $\cB = [\vcB]$. Being class functions, any representative $\vcB\in \cB$ can be used to define $I_{\vcB}(T)$.
\begin{definition}[Real Graded Tensor Model]\label{def: tensormodel} The \emph{real graded tensor model} is the measure\footnotemark:
\begin{equation}
\begin{aligned}
d\mu[T]=e^{-S[T]}\ [dT], \quad [dT]&=\prod_{a_{\mathcal{D}}} dT^{a_1\dots a_D} \cdot
\begin{cases}
\frac{1}{(2\pi)^{  \prod_c N_c/2 } }\;, & \sum_{c=1}^{D} \abs{c} = 0 \mod 2\\
1 \;, & \sum_{c=1}^{D} \abs{c} = 1 \mod 2
\end{cases} \; ,
\\ \text{with}\quad
S[T]&= \frac{1}{2} g^{\otimes D}(T,T) + \sum_{\substack{\cB\ \text{connected,}\\ |V(\cB)|>2}} \frac{\lambda_{\cB}}{|V(\cB)|} I_{\vcB}(T) \; ,
\end{aligned}
\end{equation}
where the normalization is such that \(\int d\mu[T] =1\) for \(\lambda_{\cB}=0\ \forall \cB\).
\end{definition}
\footnotetext{We treat the measures $d\mu[T]$ as a perturbed Gaußian measures. As such we do not concern ourselves with the convergence of the various tensor and matrix integrals. As we treat the Gaußian integrals as generating functions of graphs, we will not adress such issues.}

\section{Perturbative Expansions in Terms of Colored Graphs}\label{sec: pertexp}

In this section, we compute the expectation value of invariants and show that these obey the $N\to -N$ duality graph by graph in the perturbative expansion. On a technical level, most statements are generalizations of the known results for colored random tensor models, as summarized e.g.~in~\cite{gurau,tanasabook}.
Let us recall the commutation relation of the tensor component:
\begin{equation}
    T_{a_\cD} T_{b_\cD} = (-1)^{\sum_{c \in \cD } |c| }T_{a_\cD} T_{b_\cD}\;,
\end{equation}
and introduce the following short hand notations:
\begin{equation}
\begin{aligned}
    \delta^{a_\cD}_{b_\cD} = \prod_{i=1}^D \delta^{a_i}_{b_i}\;,\quad g_{a_\cD b_\cD} &= \prod_{c \in \cD} g^c_{a_c b_c} \;, \quad g^{a_\cD b_\cD} = \prod_{c \in \cD} g^{c,\, a_c b_c}\;, \\
    \text{and}\quad (\partial_T, \partial_T)&=\sum_{a_\cD, b_\cD}\frac{\partial}{\partial T^{a_\cD}} g^{a_\cD b_\cD} \frac{\partial}{\partial T^{b_\cD}} \;.
\end{aligned}
\end{equation}
We first compute the expectation values with respect to the Gaußian measure that is obtained by setting all coupling constants $\lambda_{\cB}$ to zero. This is Wick's theorem for the Gaußian expectation value $\langle\, \dots \rangle_0$ of $2k$ anticommuting variables.

\begin{lem}\label{lem: wick}
The Gaußian expectation value of an even number of tensors:
\begin{equation}
\langle T^{a^1_\cD} \dots T^{a^{2k}_\cD} \rangle_{0}
= \int [dT] e^{-\frac{1}{2} g^{\otimes D}(T,T) }\ T^{a^1_\cD} \dots T^{a^{2k}_\cD}\;,
\end{equation}
can be computed as a sum over the set $\cP_{2k}$ of (undirected) pairings of $2k$ elements:
\begin{equation}\label{eq: wick}
\langle T^{a^1_\cD} \dots T^{a^{2k}_\cD} \rangle_{0} = \sum_{P \in {\cP}_{2k}} \epsilon( \vP_{ref,2k},\vP)^{\sum_{c \in \cD }|c|} \Big( \prod_{(i,j) \in \vP} g^{a^i_\cD a^j_\cD} \Big) \;,
\end{equation}
where $\vP$ is an (arbitrarily chosen) directed version of $P$. Each summand is independent of that choice, because the sign $\epsilon( \vP_{ref,2k},\vP)$ is odd under reordering of pairs, while $g^{a^i_\cD a^j_\cD}$ is antisymmetric in the relevant cases. 
The odd moments vanish and the sign is trivial for commuting (bosonic) tensor components.
\end{lem}
\begin{proof}
This classical statement is proved using the derivative representation of normalized Gaußian measures:
\begin{equation}
\begin{split}
&\langle T^{a^1_\cD} \dots T^{a^{2k}_\cD} \rangle_{0}
= \Big[ e^{\frac{1}{2} (\partial_T, \partial_T)}\ T^{a^1_\cD} \dots T^{a^{2k}_\cD} \Big]_{T^{a^i_\cD}=0,\; \forall i} 
= \Big[ \sum_{n\geq 0} \frac{1}{n!2^n} \big((\partial_T, \partial_T)\big)^{n} T^{a^1_\cD} \dots T^{a^{2k}_\cD} \Big]_{T=0}
\\
&= \Big[ \sum_{n\geq 0} \frac{1}{n!2^n} \big((\partial_T, \partial_T)\big)^{n-1}\sum_{b_\cD, c_\cD} \frac{\partial}{\partial T^{b_\cD}} g^{b_\cD c_\cD} \sum_{r=1}^{2k}\; (-1)^{r \sum_{c \in \cD }|c|}\; \delta_{c_\cD}^{a^r_\cD}\; T^{a^1_\cD}\dots \widehat{\kern.4em T^{a^r_\cD}} \dots T^{a^{2k}_\cD} \Big]_{T=0} \;.
\end{split}
\end{equation}
Iterating this procedure, the derivatives will create (directed) pairings $\vP$ of the $2k$ tensors, and pick up minus signs if they have to anticommute with an odd number of tensors. The total sign, generated in this way, is just the sign of $\vP$ relative to the reference pairing $\vP_{ref,2k}=\{(1,2),(2,3),\dots ,(2k-1,2k)\}$.
\begin{align}
 \sum_{\vP \in \vec{\cP}_{2k}} \frac{1}{2^k}\; \epsilon( \vP_{ref,2k},\vP)^{\sum_{c \in \cD }|c|} \Big( \prod_{(i,j) \in \vP} g^{a^i_\cD a^j_\cD} \Big) \;,
\end{align}
where $\vec{\cP}_{2k}$ is the set of directed pairings of $2k$ elements. Because the sign $\epsilon( \vP_{ref,2k},\vP)$ is odd under reordering of pairs, while $g^{a^i_\cD a^j_\cD}=(-1)^{\sum_{c\in\cD}|c|}g^{a^j_\cD a^i_\cD}$, each summand does not depend on the order of the pairs in $\vP$. If $P$ is an (undirected) pairing of $2k$ elements, there are $2^k$ directed pairings $\vP$ associated to it. Taking this multiplicity into account, the statement follows.
\end{proof}
\begin{prop}\label{prop: Gaussianexpofinv}
The Gaußian expectation of an invariant of order $2k$, specified by a $D$-colored directed graph $\vcB$:
\begin{equation}
\langle I_{\vcB}(T) \rangle_0 = \int [dT] e^{-\frac{1}{2} g^{\otimes D}(T,T) }\ I_{\vcB}(T) \;,
 \end{equation}
can be computed as a sum over $(D+1)$-colored undirected graphs $\cG = [\vcG]$ (Feynman graphs) having edges of an additional color 0, such that $\cB\subset\cG$ is the maximal $D$-colored subgraph of colors $c\in\cD$:
\begin{equation}
\langle I_{\vcB}(T) \rangle_0 = \sum_{\substack{\cG,\; \cB \subset \cG \\ |V(\cG)|=2k}} \prod_{c\in\cD} \big((-1)^{|c|}N_c\big)^{F_{c/0}(\cG)} \;.
\end{equation}
The power of $N_c$ is given by the number of faces of $\cG$, of alternating colors $c$ and 0. 
\end{prop}
Since the invariants are class functions, the result does only depend on undirected graphs.
\begin{proof}
Using Lemma~\ref{lem: wick}, one has:
\begin{equation}
\begin{aligned}
\langle I_{\vcB}(T)\rangle_0 =   \sum_{\{a_\cD\}} \big\langle \prod_{(i,j) \in \vP_{ref,2k} }T^{a^i_\cD} T^{a^j_\cD} \big\rangle_0 \prod_{c\in\cD} \Big(\epsilon(\vP_{ref,2k}, \vE^c(\vcB))^{|c|} \prod_{(i,j) \in \vE^c(\vcB)}
 g_{a^i a^j}^c\Big)&
 \\
= \sum_{\{a_\cD\}} \sum_{P \in {\cP}_{2k}} \epsilon( \vP_{ref,2k},\vP)^{\sum_{c \in \cD }|c|} \Big( \prod_{(i,j) \in \vP} g^{a^i_\cD a^j_\cD} \Big)  \prod_{c\in\cD} \Big(\epsilon(\vP_{ref,2k}, \vE^c(\vcB))^{|c|} \prod_{(i,j) \in \vE^c(\vcB)}
 g^c_{a^i_c a^j_c}\Big)&\;.
\end{aligned}
\end{equation}
Now, using Property~\eqref{prop2} of the sign of two pairings to eliminate the dependence on the reference $\vP_{ref,2k}$ and reorganizing the products according to color leads to:
\begin{align}
\sum_{\{a_\cD\}} \sum_{P \in {\cP}_{2k}} 
\prod_{c\in\cD}\Bigg( \epsilon(\vP,\vE^c(\vcB))^{|c|} \Big( \prod_{(i,j) \in \vP} g^{c,\, a^i_c a^j_c} \Big)\Big( \prod_{(i,j) \in \vE^c(\vcB)} g^c_{a^i_c a^j_c}\Big)\Bigg) \;.
\end{align}
By adding edges of a new color 0 to $\vcB$, according to the pairing $\vP$, a $(D+1)$ directed colored graph $\vcG$ is obtained. Along the faces of colors $c/0$, $g^c$ and $(g^c)^{-1}$ alternate and since all indices are summed, each such face contributes a factor $N_c$. Note however, that since $g^c_{a_cb_c} g^{c,\; d_cb_c} = (-1)^{|c|}\delta_{a_c}^{d_c}$, a face picks up an additional sign if an odd number of edges are pointing in the same direction around the face (such a face was called odd). With this and the expression of $\epsilon(\vP,\vE^c(\vcB))$ in terms of the number of even faces of colors $c/0$ (Lemma~\ref{lem:sign}) we obtain:
\begin{equation}
\sum_{\{a_\cD\}} \sum_{\substack{\cG,\; \cB \subset \cG \\ |V(\cG)|=2k}} 
\prod_{c\in\cD}\Bigg( \big((-1)^{F_{c/0,even}(\vcG)}\big)^{|c|} \big((-1)^{F_{c/0,odd}(\vcG)}\big)^{|c|}  \big(N_c\big)^{F_{c/0}(\vcG)}\Bigg) \;,
\end{equation}
and since the sum of the number of even and odd face of colors $c/0$ is equal to the total number of such faces in $\vcG$ (equivalent to the ones in $\cG$), this concludes the proof. Note however, that the result does not depend on the particular choices of directed graph. This is true at each intermediate step because the symmetry properties of $\epsilon(\cdot\,,\cdot)^{|c|}$ always agree with the graded symmetry of $g^c$.
\end{proof}

At this point we have all the ingredients for the perturbative evaluation of the partition function.
\begin{prop}\label{prop: Z}
The partition function $Z(\{\lambda\}$ of the graded colored tensor model (Def.~\ref{def: tensormodel}) can be evaluated by a formal power series in the coupling constants that is indexed by $(D+1)$-colored undirected graphs $\cG$. Let us denote the $D$-colored maximally connected subgraphs of colors $c\in\cD$, called bubbles, by $\cB \subset \cG$. The sum runs only over graphs $\cG$ without bubbles having exactly two vertices:
\begin{equation}
Z(\{\lambda_\cB\}) = \sum_{\substack{\cG\\ |V(\cB)|\neq 2\; \forall \cB \subset \cG }} \frac{1}{n_b(\cG)!} \Big(\prod_{\cB \subset \cG} \frac{\lambda_\cB}{|V(\cB)|}\Big)\Big( \prod_{c\in\cD} \big((-1)^{|c|}N_c\big)^{F_{c/0}(\cG)} \Big) \;,
\end{equation}
and $n_b(\cG)$ denotes the total number of bubbles in $\cG$.
\end{prop}
\begin{proof}
We expand the interaction part to find:
\begin{align}
Z = \int  [dT]  e^{-\frac{1}{2} g^{\otimes}(T,T) } \sum_{\{p_\cB \geq 0\}} \prod_{\cB} \frac{1}{p_\cB!} \Big( \frac{\lambda_\cB}{|V(\cB)|} I_{\cB}(T)  \Big)^{p_\cB} \;,
\end{align}
where we associate to each undirected $D$-colored graph $\cB$ a multiplicity $p_\cB\geq 0$, and sum over these.
Commuting the sum and the Gaußian integral, we obtain the perturbative series:
\begin{equation}
\sum_{\{p_\cB \geq 0\}} \Bigg(\prod_{\cB} \frac{1}{p_\cB!} \Big(\frac{\lambda_\cB}{|V(\cB)|} \Big)^{p_\cB}\Bigg) \big\langle \prod_{\cB} \big(I_\cB(T)\big)^{p_\cB} \big\rangle_0 \;.
\end{equation}
As any product of invariants can be seen as a single disconnected invariant \eqref{eq: unionofinv}, Prop.~\ref{prop: Gaussianexpofinv} can be directly applied, and one obtains a sum over $(D+1)$-colored undirected graphs $\cG$ with the only condition, that they don't have bubbles $\cB \subset \cG$ with exactly two vertices. Whenever the graph $\cG$ contains a bubble $\cB$ this contributes a factor $\lambda_\cB$:
\begin{equation}
\sum_{\substack{\cG \\ |V(\cB)|\neq 2\; \forall \cB \subset \cG }} \frac{1}{n_b(\cG)!} \Big(\prod_{\cB \subset\cG } \frac{\lambda_\cB}{|V(\cB)|}\Big)\Big( \prod_{c\in\cD} \big((-1)^{|c|}N_c\big)^{F_{c/0}(\cG)} \Big) \;,
\end{equation}
here $n_b(\cG)$ is the total number of bubbles in $\cG$.
\end{proof}

\begin{cor}\label{cor: I}
The expectation value of trace invariants are computed as derivatives of the logarithm of the partition function:
\begin{equation}
\begin{aligned}
\text{for }|V(\cB)|=2:&\quad \langle g^{\otimes D}(T,T)\rangle = \Big(\prod_{c\in\cD} (-1)^{|c|}N_c\Big) + \Big(\sum_{\substack{\cB\ \text{connected,}\\ |V(\cB)|>2}} \lambda_\cB |V(\cB)|\; \frac{\del}{\del \lambda_\cB}\Big) \ln Z(\{\lambda\}) \;, \\
\text{for }|V(\cB)|>2:&\quad \langle I_{\vcB}(T)\rangle =-|V(\cB)|\; \frac{\del}{\del \lambda_\cB} \ln Z(\{\lambda\}) \;.
\end{aligned}
\end{equation}
The expectation value can be computed explicitly as a formal sum, analogous to $Z(\{\lambda\})$. The derivative acts on one $\lambda_\cB$ in the product and marks the corresponding bubble $[\cB]\subset[\cG]$:
\begin{align}
\text{for }|V(\cB)|>2:&\quad \langle I_{\vcB}(T)\rangle = \sum_{\substack{\cG\;\text{connected}, \\ \cB\subset \cG\;\text{marked},\\ |V(\cB')|\neq 2\; \forall \cB' \subset \cG }} \frac{1}{n_b(\cG)!} \Big(\prod_{\substack{ \cB' \subset \cG \\ \cB'\neq\cB}} \frac{\lambda_{\cB'}}{|V(\cB')|}\Big)\Big( \prod_{c\in\cD} \big((-1)^{|c|}N_c\big)^{F_{c/0}(\cG)} \Big) \;.
\end{align}
\end{cor}
\begin{proof}
The statements for $|V(\cB)|> 2$ follow from Def.~\ref{def: tensormodel}, \eqref{eq: ZandI}, Prop.~\ref{prop: Z} and the usual fact, that the logarithm restricts to a sum over connected graphs. For $|V(\cB)|=2$ consider the following Schwinger-Dyson-Equation:
\begin{equation}
\begin{aligned}
0&=\frac{(-1)^{\sum_c |c|}}{Z} \int [dT] \sum_{a_\cD} \frac{\del}{\del T^{a_\cD}}\Big(T^{a_\cD} e^{-S[T]}\Big)
\\
&= \Big(\prod_{c\in\cD} (-1)^{|c|}N_c\Big) - \langle g^{\otimes D}(T,T)\rangle - \;\smashoperator{ \sum_{\substack{\cB \ \text{conn.}\\ |V(\cB)|>2}} }\; \lambda_\cB \langle I_{\vcB}(T)\rangle \;,
\end{aligned}
\end{equation}
and use the result for $\langle I_{\vcB}(T)\rangle$. The first term in the second line can be interpreted as a ring-graph without any vertex but $D$ colored edges.
\end{proof}


\section{Conclusion and Outlook}\label{sec: conclusion}

In this paper, we showed the $N$ to $-N$ duality between colored tensor models with orthogonal and symplectic symmetry. We showed that the amplitude of any graph $\mathcal{G}$ comes with a factor $\prod_{c \in \mathcal{C}} \left( (-1)^{|c|} N_c \right)^{F_{c/0}(\mathcal{G})}$, where $|c|$ is a grading assigned to the color $c$ controlling its symmetry properties ($|c|=0$ for $O(N_c)$ and $|c|=1$ for $Sp(N_c)$), and $F_{c/0}(\mathcal{G})$ is the number of faces of color $c$ in $\mathcal{G}$. As a consequence, the amplitude of any graph is invariant under simultaneously changing $O(N_c) \leftrightarrow Sp(N_c)$ and $N_c \leftrightarrow - N_c$ for any color $c$. Our analysis relies heavily on the properties of the sign of oriented pairings. This sign ensures that the invariants are specified by undirected graphs and do not depend on the arbitrarily chosen orientation of edges. Then, by carefully keeping track of the signs coming from the bosonic or fermionic nature, as well as reorientations of edges in directed graphs, we obtain the desired expressions for the graph amplitudes.

A natural follow up of this project is to consider tensors with symmetry properties under permutation of their indices. The work of \cite{King,Cvitanovic,Mkrtchyan-Veselov,Cvitanovicbook} suggests, that similar models could be constructed for arbitrary (also mixed) tensor representations of the orthogonal and symmetric group. E.g.~we suspect a relation between a tensor model with tensors of order $D$ transforming in the totally symmetric representation of $O(N)$ (similar to \cite{sylvan,Carrozza:2021qos,Carrozza}) and a model build on the totally antisymmetric representation of $Sp(N)$. For odd $D$ the symplectic model should again be fermionic. Moreover, one should also explore  the implications of the $N\to -N$ duality for tensor field theories. The sign change may generate new renormalization group fixed points. Finally, more general models with $OSp(m,n)$ symmetry could be considered. In this case one would interchange $n$ and $m$.

\paragraph*{Acknowledgements.}
We thank Razvan Gurau, Thomas Krajewski and Adrian Tanasa for comments and discussions. The authors gratefully acknowledge support of the Institut Henri Poincaré (UAR 839 CNRS-Sorbonne Université), and LabEx CARMIN (ANR-10-LABX-59-01), where some of this work has been carried out. 
H.~K.~has been supported by the Deutsche Forschungsgemeinschaft (DFG, German Research Foundation) under Germany's Excellence Strategy EXC--2181/1 -- 390900948 (the Heidelberg STRUCTURES Cluster of Excellence).
T.~M.~has been partially supported by the ANR-20-CE48-0018 “3DMaps” grant.

{\pagestyle{plain}
	\bibliography{references} 
	\addcontentsline{toc}{section}{References}
}


\end{document}